\documentclass[11pt]{article}
\usepackage{amssymb,amsmath,amsthm,amscd,latexsym}
\usepackage{mathrsfs,color}
\usepackage{mathrsfs}
\usepackage{amsfonts}
\usepackage{amssymb}
\usepackage[all]{xy}
\usepackage{makecell}
\usepackage{xypic}
\usepackage{graphicx}
\usepackage{epstopdf}
\usepackage{booktabs}
\usepackage{geometry}
\usepackage{microtype}
\usepackage[justification=centering]{caption}
\usepackage{array, caption, threeparttable}
\usepackage[font=small,labelfont=bf,labelsep=none]{caption}

\makeatletter
\long\def\@makecaption#1#2{%
  \vskip\abovecaptionskip
   \sbox\@tempboxa{#1:#2}%
  \ifdim \wd\@tempboxa >\hsize
   {\bfseries #1:} #2\par
  \else
    \global \@minipagefalse
   \hb@xt@\hsize{\box\@tempboxa\hfil}%
  \fi
  \vskip\belowcaptionskip}
\makeatother

\renewcommand{\paragraph}{\roman{paragraph}}
\input cyracc.def

\newcommand{\Z}{\mathbb{Z}}
\newcommand{\F}{\mathbb{F}}

\newtheorem{thm}{Theorem}[section]
\newtheorem{lem}[thm]{\scshape  Lemma}
\newtheorem{coro}[thm]{\scshape  Corollary}
\newtheorem{prop}[thm]{\scshape  Proposition}
\newtheorem{defe}[thm]{\scshape Definition}
\newtheorem{exa}[thm]{\scshape   Example}
\newtheorem{rem}[thm]{\scshape   Remark}

\begin{document}

\title{\bf Balanced reconstruction codes for single edits
\thanks{The research of X. Zhang was supported by the NSFC under Grants No. 12171452 and No. 11771419, the Anhui Initiative in Quantum
Information Technologies under Grant No. AHY150200,  the National Key Research and Development Program of China (2020YFA0713100), and the Innovation Program for Quantum Science and Technology (2021ZD0302904).
The research of R. Wu was supported by China Postdoctoral Science Foundation under Grant No. 2021M703098.}
}
\author{
Rongsheng Wu\thanks{School of Mathematical Sciences, University of Science and Technology of China, Hefei 230026, Anhui, China (e-mail:
wrs2510@ustc.edu.cn)},
and
Xiande Zhang\thanks{CAS Wu Wen-Tsun Key Laboratory of Mathematics, School of Mathematical Sciences, University of Science and Technology of China, Hefei, Anhui, 230026,
P.R. China (e-mail: drzhangx@ustc.edu.cn)}\ \thanks{Hefei National Laboratory, University of Science and Technology of China, Hefei 230088, China}
}
\date{}
\maketitle
{\bf Abstract:} {Motivated by the sequence reconstruction problem initiated by Levenshtein,  reconstruction codes were introduced by Cai \emph{et al}. to combat errors when a fixed number of noisy channels are available. The central problem on this topic is to design codes with sizes as large as possible, such that  every codeword can be uniquely reconstructed from any $N$ distinct noisy reads, where $N$ is fixed. In this paper, we study binary reconstruction codes with the constraint that
every codeword is balanced, which is a common requirement in the technique of DNA-based storage. For all possible channels with a single edit error and their variants, we design asymptotically optimal balanced reconstruction codes for all $N$, and show that the number of their redundant symbols decreases from $\frac{3}{2}\log _2 n+O(1)$ to $\frac{1}{2}\log_2n+\log_2\log_2n+O(1)$, and finally to $\frac{1}{2}\log_2n+O(1)$ but with different speeds, where $n$ is the length of the code. Compared with the unbalanced case, our results imply that the balanced property does not reduce the rate of the reconstruction code in the corresponding codebook.

}


\section{Introduction}\label{Sec:1}

The sequence reconstruction problem has been extensively studied in the literature by many researchers since 2001 due to Levenshtein \cite{L01,L01a}. The original motivation was to combat errors by repeatedly transmitting a message without coding in situations when no other method is feasible. One of the central problems in this area is to determine the necessary number of transmissions for an arbitrary message, or equivalently, the maximum intersection size between the error-balls of two different words in a codebook. Each transmission is referred to as an independent noisy channel.
Levenshtein \cite{L01,L01a} addressed this problem for combinatorial channels with
several types of errors of most interest in the field of coding theory,
 namely, substitutions, insertions and deletions. Later, much work has been done concerning the sequence reconstruction problems for different error models, such as signed permutations distorted by reversal errors \cite{E05,E08}, and general error
graphs \cite {LS09}.


\par
Note that for
all the works mentioned above, the transmitted sequences are selected from the entire space without coding. Recently due to applications in DNA-based storage, the sequence reconstruction problem was studied under the setting where the transmitted sequences are chosen from a given code with a certain error-correcting property.
For example in \cite{YSLB13}, permutation codes with prescribed minimum Kendall's $\tau$ distances $1,2$ and $2r$ were considered. Gabrys and Yaakobi \cite{GY16,GY18} studied the  channels causing $t$ deletions where the transmitted sequences belong to a binary single-deletion-correcting code.  In 2017, Sala \emph{et al.} \cite{SGSD17} studied the insertion channels where the transmitted sequences have pairwise edit distance at least $2l$, for any $l\ge 0$, which  generalizes the results of Levenshtein in \cite{L01,L01a}.

\par

Considering a fixed number of erroneous channels during the sequencing process of a DNA strand, Cai \emph{et al.} \cite{CKNY21} (see also \cite{CKY20}) proposed the dual problem of the sequence reconstruction as follows. In most sequencing platforms, multiple copies of the same DNA strand are created after undergoing polymerase chain reaction (PCR). The sequencer reads all copies and provides many possibly inaccurate reads to the user, who then needs to further
reconstruct the original DNA strand from these noisy reads.  When a fixed number of distinct noisy reads are provided, the main task is to design a codebook such that every codeword can be uniquely reconstructible from these distinct noisy reads. This problem has a quite different flavor from the original reconstruction problem, but can be viewed as an extension of the classical error-correcting codes. Leveraging on these multiple channels (or reads), one can increase the information capacity, or equivalently, reduce the number of redundant bits for these next-generation devices. In \cite{CKNY21}, Cai \emph{et al.} almost completely determined the asymptotic optimal redundancy of the code when the channels are affected by a single edit. Chrisnata \emph{et al}. \cite{CK21,CKY22} extended the case for $t$ deletions, and provided an explicit code that is uniquely reconstructible with certain parameters in the two-deletion channel.

In this paper, we follow the framework initiated by Cai \emph{et al.} \cite{CKNY21}, to study the so-called \emph{reconstruction codes} with additional constraints required in DNA-based storage technique. The first interesting constraint is the balanced property of sequences. It has been shown that binary
balanced error-correcting codes play a significant role in constructing GC-balanced error-correcting codes [35], which are widely used in the DNA coding theory [28, 36] since they are more
stable than unbalanced DNA strands and have better coverage during sequencing. Further,
it is well known that balanced codes are DC-free \cite{DH88} and have attractive applications in the encoding of unchangeable data on a laser disk \cite{K86,L84}.
Much efforts have been devoted to constructing binary balanced error-correcting codes in the literature \cite{AB93,FW02,TB89}.

By the above considerations, this paper focuses on the study of binary reconstruction codes able to uniquely recover a \emph{balanced} sequence from a fixed number of  erroneous channels affected by a single edit (a substitution, deletion,
or insertion) and its variants. For all related errors, we determine the optimal redundancy and construct asymptotically optimal codes. In particular, these results show that the balanced property does not reduce the ratio of the reconstruction code to the corresponding codebook compared to the unbalanced one.


The rest of this paper is organized as follows. Section \ref{Sec:2} introduces the main notation and provides the necessary background needed in the subsequent sections. In addition, some sufficient and necessary conditions for the intersection size of error-balls are provided. Then Sections \ref{Sec:3}-\ref{Sec:4} are devoted to characterizing the asymptotic optimal redundancy of a balanced $(n,N;B_2)$-reconstruction code with $B_2\in\{B^{\rm D},B^{\rm I},B^{\rm DI}\}$ and $B_2\in\{B^{\rm SD},B^{\rm SI},B^{\rm edit}\}$, respectively. Finally, we conclude this paper in Section 5.

\section{Preliminaries}\label{Sec:2}
Let $\F_2$ denote the binary alphabet $\{0,1\}$, and let  $\F_2^n$ denote the set of all binary sequences of length $n$. The Hamming weight of $\textbf{x}\in \F_2^n$, denoted by ${\rm wt_H}(\textbf{x})$, is the number of indices $i$ where $x_i\neq 0$, and the
Hamming distance $d_{\rm H}(\textbf{x},\textbf{y})$ between two words $\textbf{x},\textbf{y}\in \F_2^n$ is defined to be the number
of coordinates in which $\textbf{x}$ and $\textbf{y}$ differ.
Assume that $n$ is even throughout this paper for convenience. A word in $\F_2^n$ is \emph{balanced} if it has exactly $n/2$ ones. Let $U_n$ be the set of all balanced words in $\F_2^n$. A balanced code is a subset of $U_n$.

We introduce the concept of reconstruction codes as in \cite{CKNY21}.
First, we define the following seven error-ball functions for $\textbf{x}\in \F_2^n$. Let $B^{\rm S}(\textbf{x})$, $B^{\rm D}(\textbf{x})$ and $B^{\rm I}(\textbf{x})$ denote the set of all words obtained from $\textbf{x}$ via at most one substitution, one deletion, and one insertion, respectively.  Combining these functions, we define further that
$$B^{\rm DI}(\textbf{x}):=B^{\rm D}(\textbf{x})\cup B^{\rm I}(\textbf{x}),\ \ \ B^{\rm SD}(\textbf{x}):=B^{\rm S}(\textbf{x})\cup B^{\rm D}(\textbf{x}),$$
$$B^{\rm SI}(\textbf{x}):=B^{\rm S}(\textbf{x})\cup B^{\rm I}(\textbf{x}),\ \ \ B^{\rm edit}(\textbf{x}):=B^{\rm S}(\textbf{x})\cup B^{\rm D}(\textbf{x})\cup B^{\rm I}(\textbf{x}).$$
 For example, let $\textbf{x}=1010 \in U_4$. Then $B^{\rm S}(\textbf{x})=\{1010,0010,1110,1000,1011\}\subseteq \F_2^4$, $B^{\rm D}(\textbf{x})=\{010,110,100,101\}\subseteq \F_2^3$, $B^{\rm I}(\textbf{x})=\{01010,11010,10010,10110,10100,10101\}\subseteq \F_2^5$, and the remaining error-balls for $\textbf{x}$ are the corresponding union between $B^{\rm S}(\textbf{x})$, $B^{\rm D}(\textbf{x})$ and $B^{\rm I}(\textbf{x})$ above.

Let $B_2$ be the noisy channel corresponding to any one of the above functions, that is, $B_2\in \{B^{\rm S}, B^{\rm D}, B^{\rm I}, B^{\rm DI}, B^{\rm SD}, B^{\rm SI},B^{\rm edit}\}$.
For any $\mathcal{C}\subseteq \F_2^n$, the \emph{read coverage} of $\mathcal{C}$ for channel $B_2$, denoted by $\nu(\mathcal{C};B_2)$, is defined to be the maximum intersection size between error-balls of any two different codewords in $\mathcal{C}$. More specifically,
$$\nu(\mathcal{C};B_2)=\max\{|B_2(\textbf{x})\cap B_2(\textbf{y})|:\textbf{x},\textbf{y}\in \mathcal{C}\ {\rm and}\ \textbf{x}\neq \textbf{y}\}.$$
The quantity $\nu(\mathcal{C};B_2)$ was introduced by Levenshtein \cite{L01}, who showed that the number of channels
required to reconstruct a codeword from $\mathcal{C}$ is at least $\nu(\mathcal{C};B_2)+1$. The problem of determining $\nu(\mathcal{C};B_2)$ is referred to as the sequence reconstruction
problem.

For a fixed constant $N$, if a code $\mathcal{C}\subseteq \F_2^n$ satisfies  $\nu(\mathcal{C};B_2)<N$, then we call $\mathcal{C}$  an $(n,N;B_2)$-reconstruction code. As mentioned in Introduction, we focus on balanced reconstruction codes in this paper. The fundamental problem on this topic is to estimate the minimum number of redundant bits for such a code. Define the \emph{redundancy} of a code $\mathcal{C}\subseteq \F_2^n$ to be the value $n-\log_2 |\mathcal{C}|$. Then we are interested in studying the following quantity,
$$\rho_b(n,N;B_2)=\min\{n-\log_2 |\mathcal{C}|: \mathcal{C}\subseteq U_n \ {\rm and} \  \nu(\mathcal{C};B_2)<N\}.$$
Note that the case $N = 1$ is the classical model which has been studied for years in the design of balanced error-correcting
codes \cite{AB93,FW02,TB89}. 

\subsection{An easy result}

We determine the optimal redundancy for the channel causing a single substitution error in this subsection, that is $\rho_{b}(n,N;B^{\rm S})$. We will apply the following useful estimation of binomials (see e.g., \cite[Proposition 3.6.2]{JJ09}) frequently.

\begin{lem}\cite{JJ09}\label{binomial}
For all even $n\geq 2$, we have
$$\frac{2^n}{\sqrt{2n}} \leq \binom{n}{n/2} \leq \frac{2^n}{\sqrt{n}}.$$
\end{lem}

The case $N=1$ is  briefly explained in the next example, which is equivalent to the design of classical binary balanced error-correcting codes.

\begin{exa}\label{substitution1}
Let $\mathcal{C}$ be a balanced $(n,1;B^{\rm S})$-reconstruction code with maximum size, then $\mathcal{C}$ is a binary code of length $n$ with minimum Hamming distance $4$ and constant weight $n/2$. By the lower bound on the size of constant weight codes discovered by Graham and Sloane \cite[Theorem 1]{GS80}, we have
\begin{eqnarray*}
 |\mathcal{C}| &\geq& \frac{1}{n}\binom{n}{n/2} \geq \frac{2^{n-1/2}}{n^{3/2}},
\end{eqnarray*}
where the last inequality follows from  Lemma \ref{binomial}.
The upper bound is due to Agrell, Vardy and Zeger \cite{AVZ00} by letting $w=n/2,\ d=4$ and $ t=n/2-1$ in \cite[Theorem 12]{AVZ00},
\begin{eqnarray*}
 |\mathcal{C}| \leq \binom{n}{n/2-1}\bigg/\binom{n/2}{n/2-1} &\leq& \frac{2^{n+1}}{n^{3/2}},
\end{eqnarray*}
where the last inequality follows from Lemma \ref{binomial} again.
\end{exa}

A simple application of \cite[Corollary 1]{L01} gives the following lemma.

\begin{lem}\label{Scondition}
Let $\textbf{x}$ and $\textbf{y}$ be different words in $U_n$. Then
\[
B^{\rm S}(\textbf{x})\cap B^{\rm S}(\textbf{y})=\begin{cases}
2, & {\rm if}\ d_{\rm H}(\textbf{x},\textbf{y})=2, \\
0, & {\rm if}\ d_{\rm H}(\textbf{x},\textbf{y})\geq 4.
\end{cases}
\]
\end{lem}

Then it is ready to determine the optimal redundancy for the error-ball $B^{\rm S}$
 as follows.

\begin{thm}\label{substitution}
For the error-ball $B^{\rm S}$, we have
\[
\rho_{b}(n,N;B^{\rm S})=\begin{cases}
\frac{3}{2}\log_2n+\Theta(1), & N\in \{1,2\}, \\
\Delta, & N\geq 3,
\end{cases}
\]where $\Delta:= n-\log_2\binom{n}{n/2}=\frac{1}{2}\log n+\Theta(1)$ is the redundancy of $U_n$.
\end{thm}
\begin{proof}
The value $\rho_{b}(n,1;B^{\rm S})$ follows immediately from Example \ref{substitution1}. For the case $N\geq 2$, we clearly have $\rho_{b}(n,1;B^{\rm S})=\rho_{b}(n,2;B^{\rm S})$ from Lemma \ref{Scondition} and the fact that $\nu(U_n;B^{\rm S})=\Delta$.
\end{proof}

\subsection{The intersection size of various error-balls}

This subsection deals with the intersection size between error-balls of any two different words in $U_n$. We need the following notion of confusability which was introduced in \cite{CKNY21} for general $q$-ary words. Here, we restrict the definition to balanced words.


\begin{defe}\label{confusableAB}Suppose $\textbf{x}=\textbf{u}\textbf{c}\textbf{v} $ and $  \textbf{y}=\textbf{u}\textbf{c}'\textbf{v}$ are two distinct words in $U_n$ for some subwords $\textbf{u}$, $\textbf{v}$, $\textbf{c}$ and $\textbf{c}'$. We say that $\textbf{x}$ and $\textbf{y}$ are
\begin{enumerate}
  \item  \emph{Type-A-confusable with $m$} if
$\{\textbf{c},\textbf{c}'\}$ is of the form $\{(10)^m,(01)^m\}$ for $m\geq 1$; and
  \item  \emph{Type-B-confusable with $m$} if $\{\textbf{c},\textbf{c}'\}$ is either the form $\{01^m,1^m0\}$ or $\{10^m,0^m1\}$ for $m\geq 2$.
\end{enumerate}
\end{defe}

\begin{exa}
Let $\textbf{x}=11101000, \ \textbf{y}=11010100\in U_8$. Then $\textbf{x}$ and $\textbf{y}$ are Type-A-confusable with $m=2$, $\textbf{u}=11$ and $\textbf{v}=00$. Similarly, $\textbf{x}'=111000$ and $\textbf{y}'=101100$ are Type-B-confusable with $m=2$, $\textbf{u}=1$ and $\textbf{v}=0$.

\end{exa}

Let $\textbf{x}$ and $\textbf{y}$ be two distinct words in $U_n$. For $B_2\in\{B^{\rm D},B^{\rm I}\}$, we know that $B_2(\textbf{x})\cap B_2(\textbf{y})\leq 2$ by observations in \cite{L01} or \cite{L01a}, and therefore we have $\rho_{b}(n,N;B_2)=\Delta$ for $N\geq 3$. Next, we characterize the intersection sizes of error-balls for different channels as in  \cite{CKNY21}.

\begin{prop}\label{D-Icondition}
Let $B_2\in \{B^{\rm D},B^{\rm I}\}$, and let $\textbf{x}, \textbf{y}$ be two distinct words in $U_n$.
\begin{enumerate}
    \item[\rm (i)] If $d_{\rm H}(\textbf{x},\textbf{y})=2$, then $|B_2(\textbf{x})\cap B_2(\textbf{y})|=1$ if and only if $\textbf{x}$ and $\textbf{y}$ are Type-B-confusable.
    \item[\rm (ii)] $|B_2(\textbf{x})\cap B_2(\textbf{y})|=2$ if and only if $\textbf{x}$ and $\textbf{y}$ are Type-A-confusable.
    \item[\rm (iii)] $|B^{\rm D}(\textbf{x})\cap B^{\rm D}(\textbf{y})|=|B^{\rm I}(\textbf{x})\cap B^{\rm I}(\textbf{y})|$.
\end{enumerate}
\end{prop}
\begin{proof}
Since the words $\textbf{x}$ and $\textbf{y}$ belong to $U_n$ with the same Hamming weight $\frac{n}{2}$, parts (i) and (ii) are true according to \cite[Propositions 9 and 12]{CKNY21}. The rest case (iii) then follows from part (ii) and the fact that
$|B^{\rm D}(\textbf{x})\cap B^{\rm D}(\textbf{y})|=0$ if and only if $|B^{\rm I}(\textbf{x})\cap B^{\rm I}(\textbf{y})|=0.$
\end{proof}

A corollary of the above result is immediate.

\begin{coro}\label{DIcondition}
Let $\textbf{x}$ and $\textbf{y}$ be two distinct words in $U_n$. Then $|B^{\rm DI}(\textbf{x})\cap B^{\rm DI}(\textbf{y})|\in\{0,2,4\}$. In particular, $|B^{\rm DI}(\textbf{x})\cap B^{\rm DI}(\textbf{y})|=4$ if and only if $\textbf{x}$ and $\textbf{y}$ are Type-A-confusable. Moreover, we have $\rho_{b}(n,N;B_2^{\rm DI})=\Delta$ for $N\geq 5$.
\end{coro}

Combining Lemma \ref{Scondition}, Proposition \ref{D-Icondition} and Corollary \ref{DIcondition}, we have the following two propositions for the intersection size of the error-balls which involve substitutions. The proof is straightforward and thus omitted.

\begin{prop}\label{SDIcondition}
Let $B_2\in \{B^{\rm SD},B^{\rm SI}\}$, and let $\textbf{x}, \textbf{y}$ be two distinct words in $U_n$.
\begin{enumerate}
    \item[\rm (i)] If $d_{\rm H}(\textbf{x},\textbf{y})=2$, then $|B_2(\textbf{x})\cap B_2(\textbf{y})|\in\{2,3,4\}$. In particular, $|B_2(\textbf{x})\cap B_2(\textbf{y})|=4$ if and only if $\textbf{x}$ and $\textbf{y}$ are Type-A-confusable with $m=1$; and $|B_2(\textbf{x})\cap B_2(\textbf{y})|=3$ if and only if $\textbf{x}$ and $\textbf{y}$ are Type-B-confusable.
    \item[\rm (ii)] If $d_{\rm H}(\textbf{x},\textbf{y})\geq 4$, then $|B_2(\textbf{x})\cap B_2(\textbf{y})|\leq 2$. In particular, $|B_2(\textbf{x})\cap B_2(\textbf{y})|=2$ if and only if $\textbf{x}$ and $\textbf{y}$ are Type-A-confusable with $m\geq 2$.
    \item[\rm (iii)] $|B^{\rm SD}(\textbf{x})\cap B^{\rm SD}(\textbf{y})|=|B^{\rm SI}(\textbf{x})\cap B^{\rm SI}(\textbf{y})|$.
\end{enumerate}
Moreover, we have $\rho_{b}(n,N;B_2)=\Delta$ for $N\geq 5$.
\end{prop}

\begin{prop}\label{Econdition}
Let $\textbf{x}$ and $\textbf{y}$ be two distinct words in $U_n$. Then we have $|B^{\rm edit}(\textbf{x})\cap B^{\rm edit}(\textbf{y})|\in\{0,2,4,6\}$ and the following hold.
\begin{enumerate}
    \item[\rm (i)] If $d_{\rm H}(\textbf{x},\textbf{y})=2$, then $|B^{\rm edit}(\textbf{x})\cap B^{\rm edit}(\textbf{y})|\in \{2,4,6\}$. In particular, $|B^{\rm edit}(\textbf{x})\cap B^{\rm edit}(\textbf{y})|=6$ if and only if $\textbf{x}$ and $\textbf{y}$ are Type-A-confusable with $m=1$; and $|B^{\rm edit}(\textbf{x})\cap B^{\rm edit}(\textbf{y})|=4$ if and only if $\textbf{x}$ and $\textbf{y}$ are Type-B-confusable.
    \item[\rm (ii)] If $d_{\rm H}(\textbf{x},\textbf{y})\geq 4$, then $|B^{\rm edit}(\textbf{x})\cap B^{\rm edit}(\textbf{y})|\in \{0,2,4\}$. In particular, $|B^{\rm edit}(\textbf{x})\cap B^{\rm edit}(\textbf{y})|=4$ if and only if $\textbf{x}$ and $\textbf{y}$ are Type-A-confusable with $m\geq 2$.
\end{enumerate}
Moreover, we have $\rho_{b}(n,N;B_2^{\rm edit})=\Delta$ for $N\geq 7$.
\end{prop}

The following result is an analogy to \cite[Theorem 23]{CKNY21}, which presents the lower bounds for the redundancy of the code under certain conditions with respect to the notion of confusability. The proof is similar to the non-restricted case \cite{CKY20} and thus omitted.

\begin{prop}\label{lowerbound}
Let $\mathcal{C}\subseteq U_n$. Then the following hold.
\begin{enumerate}
    \item[\rm (i)] If every pair of distinct words in $\mathcal{C}$ are not Type-A-confusable, then the redundancy of $\mathcal{C}$ is at least $\frac{1}{2}\log_2n+\log_2\log_2n-O(1)$.
    \item[\rm (ii)] If every pair of distinct words in $\mathcal{C}$ are not Type-B-confusable, then the redundancy of $\mathcal{C}$ is at least $\frac{1}{2}\log_2n+\log_2\log_2n-O(1)$.
    \item[\rm (iii)] If every pair of distinct words in $\mathcal{C}$ are not Type-B-confusable with $m=1$, then the redundancy of $\mathcal{C}$ is at least $\Delta+1-o(1)$.
\end{enumerate}
\end{prop}

\section{Reconstruction codes with error-balls $B^{\rm D}$, $B^{\rm I}$ and $B^{\rm DI}$}\label{Sec:3}

In this section, we determine the  optimal redundancy of balanced reconstruction codes for the error-balls $B^{\rm D}$, $B^{\rm I}$ and $B^{\rm DI}$. We first consider the case $N=1$.

\subsection{The case $N=1$}

It is known that any code can correct $s$ deletions if and only if it can correct $s$ insertions \cite{L65}. Thus, we only consider one of the error-balls $B^{\rm D}$ or $B^{\rm I}$ in the rest of this subsection.  Let $A^{\rm D}(n)$ denote the maximum size of a binary balanced single-deletion correcting code of length $n$. Then it suffices to estimate the value of $A^{\rm D}(n)$.

Since the Varshamov-Tenengolts (VT) codes are the best known binary codes that can correct a single deletion \cite{S00}, we define a balanced Varshamov-Tenengolts (BVT for short) code for our purpose.


\begin{defe}\label{restrictedVT}
{\rm(BVT code)} For any $0\leq a \leq n$, the balanced Varshamov-Tenengolts code $BVT_a(n)$ is defined as follows:
$$BVT_a(n)=\left\{(x_1,x_2,\ldots,x_n)\in U_n:\sum_{i=1}^nix_i \equiv a\pmod {n+1}\right\}.$$
\end{defe}
Obviously, the set $BVT_a(n)$ is a binary balanced single-deletion correcting code for any $0\leq a \leq n$. Hence, $A^{\rm D}(n)\geq \binom{n}{n/2}\big/(n+1)$ trivially.
It should be noted that, several modifications of the VT-code have previously been
proposed for different purposes, see e.g. \cite{BM19,SWGY17}.


Next, we seek an upper bound on the value of $A^{\rm D}(n)$.  We need a few preliminary results on hypergraphs, which are mainly from \cite{B89} and \cite{KK13}.
Let $X$ be a finite set. A hypergraph $\mathcal{H}=(X,H)$ on $X$ is a family  $H$ of nonempty subsets of $X$, where elements of $X$ are called vertices, and elements of $H$ are called hyperedges. 
A matching of a hypergraph is a collection of pairwise disjoint hyperedges, and the matching number of $\mathcal{H}$, denoted by $\nu(\mathcal{H})$, is the largest number of edges in a matching of $\mathcal{H}$.
A transversal of a hypergraph $\mathcal{H}=(X,H)$ is a subset $T\subset X$ that intersects every hyperedge in $H$, and the transversal number of $\mathcal{H}$, denoted by $\tau(\mathcal{H})$, is the smallest size of a transversal. Suppose that $\mathcal{H}$ has $n$ vertices and $m$ edges, let $A_{n \times m}$ be the incidence matrix of $\mathcal{H}$. Kulkarni \emph{et al}. \cite{KK13} proved that the matching number and transversal number of a hypergraph $\mathcal{H}$ are solutions of the following integer linear programming problems:
$$\nu(\mathcal{H})=\max\bigg\{\sum_{i=1}^mz_i:A\textbf{z}\leq \textbf{1}_n\bigg\}\ {\rm and} \ \ \tau(\mathcal{H})=\min\bigg\{\sum_{i=1}^nw_i:A^T\textbf{w}\geq  \textbf{1}_m\bigg\},$$
where $\textbf{z}=(z_1,z_2,\ldots,z_m)^T\in \Z_{\geq 0}^m$, $\textbf{w}=(w_1,w_2,\ldots,w_n)^T\in \Z_{\geq 0}^n$, $\textbf{1}_n$ is the all one column vector, and the inequality means the components-wise inequality. In particular, $\nu(\mathcal{H})\leq \tau(\mathcal{H})$.
If we relax the choice of $z_i$ and $w_i$ to any nonnegative reals in the above programming problems, we obtain the definitions of fractional matching number and fractional transversal number of $\mathcal{H}$, denoted by $\nu^*(\mathcal{H})$ and $\tau^*(\mathcal{H})$, respectively.


We will apply the following lemma to give an upper bound of $A^{\rm D}(n)$.

\begin{lem}\cite{KK13}\label{fractional}
For any hypergraph $\mathcal{H}$, we have
$$\nu(\mathcal{H})\leq \nu^*(\mathcal{H})=\tau^*(\mathcal{H}) \leq \tau(\mathcal{H}).$$
\end{lem}

Let $V_{n-1}$ be the subset of $\F_2^{n-1}$ consisting of all words with Hamming weights $\frac{n}{2}$ or $\frac{n}{2}-1$. Then $|V_{n-1}|=\binom{n-1}{n/2}+\binom{n-1}{n/2-1}=\binom{n}{n/2}$. Consider the following hypergraph:
$$\mathcal{H}_{n}^{\rm D}=\big(V_{n-1},\{B^{\rm D}(\textbf{x}):\textbf{x}\in U_n\}\big).$$
In $\mathcal{H}_{n}^{\rm D}$,  the vertices are words in $V_{n-1}$, and the hyperedges are single-deletion balls of words in $U_n$.
Then the value of $A^{\rm D}(n)$ is equal to the matching number $\nu(\mathcal{H}_{n}^{\rm D})$ of $\mathcal{H}_{n}^{\rm D}$.
By Lemma \ref{fractional}, we have $\nu(\mathcal{H}_{n}^{\rm D})\leq \tau^*(\mathcal{H}_{n}^{\rm D})$, where $\tau^*(\mathcal{H}_{n}^{\rm D})$ is the fractional transversal number of $\mathcal{H}_{n}^{\rm D}$. By definition,
$$\tau^*(\mathcal{H}_{n}^{\rm D})=\min\bigg\{\sum_{\textbf{x}\in V_{n-1}} w(\textbf{x}):\sum_{\textbf{x}\in B^{\rm D}(\textbf{y})}w(\textbf{x})\geq 1,\ \textbf{x}\in V_{n-1},\ \textbf{y}\in U_n\ {\rm and}\ w(\textbf{x})\geq 0\bigg\}.$$

Next, we will give an upper bound of $\tau^*(\mathcal{H}_{n}^{\rm D})$ by computing $\sum_{\textbf{x}\in V_{n-1}} w(\textbf{x})$ for the special function $w(\textbf{x})=\frac{1}{r(\textbf{x})}$, where  $r(\textbf{x})$ is the number of runs in $\textbf{x}$. Consequently, this will give an upper bound for  $A^{\rm D}(n)$. First, we need the following counting lemma.

\begin{lem}\label{Vruns}
The number of words in $U_n$ with exactly $i$ $(2 \leq i \leq n)$ runs is
$$2\cdot \binom{n/2-1}{\lceil i/2\rceil-1}\binom{n/2-1}{\lfloor i/2\rfloor-1}.$$
Further, the number of words in $V_{n-1}$ with exactly $i$ $(2 \leq i \leq n-1)$ runs is
$$2\cdot \bigg[\binom{n/2-1}{\lceil i/2\rceil-1}\binom{n/2-2}{\lfloor i/2\rfloor-1}+\binom{n/2-1}{\lfloor i/2\rfloor-1}\binom{n/2-2}{\lceil i/2\rceil-1}\bigg].$$
\end{lem}
\begin{proof}
We only prove the case for $U_n$, and a similar argument works for $V_{n-1}$.
Let $T=T_0\cup T_1$ be the subset in $U_n$ with exactly $i$ runs, where $T_i$ consists of words in $T$  with the first coordinate being $i$. Then it is easy to see that $T_0\cap T_1=\emptyset$ and $|T_0|=|T_1|$. Hence, we only need to calculate the value $|T_1|$. Let $\textbf{x}$ be a word in $T_1$, and then it is of the form:
$$\textbf{x}=(\underbrace{\textbf{1010}\ldots \textbf{a}}_i),$$
where each boldface symbol represents a run of length at least one. Moreover, $\textbf{a}=\textbf{1}$ if $i$ is odd, otherwise $\textbf{a}=\textbf{0}$. Then the size of $T_1$ equals the number of ways to distribute $\frac{n}{2}$ 1s into $\lceil\frac{i}{2}\rceil$ blocks, and distribute $\frac{n}{2}$ 0s into  $\lfloor\frac{i}{2}\rfloor$ blocks, respectively, such that each block is nonempty.
\end{proof}

\begin{thm}\label{deletion1}
Let $n\geq2$ be even. Then the maximum size of a balanced $(n,1;B^{\rm D})$-reconstruction code,  
\begin{eqnarray*}
  A^{\rm D}(n) &\leq& \frac{2\big(\binom{n}{n/2}-2\big)}{n-2}.
\end{eqnarray*}
Consequently, the optimal redundancy $\rho_{b}(n,N;B^{\rm D})$ is at least $ \frac{3}{2}\log_2n-O(1)$.
\end{thm}
\begin{proof}
As indicated in Lemma \ref{fractional}, we have $|A^{\rm D}(n)|=\nu(\mathcal{H}_{n}^{\rm D})\leq \tau^*(\mathcal{H}_{n}^{\rm D})$. Let $w(\textbf{x})=\frac{1}{r(\textbf{x})}$, where $\textbf{x}\in V_{n-1}$, and $r(\textbf{x})$ is the number of runs in $\textbf{x}$. Then for any $\textbf{y}\in U_n$,
$$\sum_{\textbf{x}\in B^{\rm D}(\textbf{y})}w(\textbf{x})\overset{(a)}{\geq} \frac{|B^{\rm D}(\textbf{y})|}{r(\textbf{y})}=1,$$
where the inequality (a) follows from $r(\textbf{x})\leq r(\textbf{y})$ (see \cite[Lemma 3.2]{KK13}).
By Lemma \ref{Vruns}, the quantity $\sum_{\textbf{x}\in V_{n-1}} w(\textbf{x})$ equals
\begin{eqnarray*}
  &&2\cdot\sum_{i=1}^{n-1}\bigg[\binom{n/2-1}{\lceil i/2\rceil-1}\binom{n/2-2}{\lfloor i/2\rfloor-1}+\binom{n/2-1}{\lfloor i/2\rfloor-1}\binom{n/2-2}{\lceil i/2\rceil-1}\bigg]\frac{1}{i}\\
   &=& \frac{4}{n-2}\cdot \sum_{i=1}^{n-1}\binom{n/2-1}{\lceil i/2\rceil-1}\binom{n/2-1}{\lfloor i/2\rfloor-1}\frac{n-i}{i}.
\end{eqnarray*}
Then the proof of the theorem needs the following combinatorial inequality.

\begin{lem}\label{estimate}
With the notation above, we have
\begin{eqnarray*}
   \sum_{i=1}^{n-1}\binom{n/2-1}{\lceil i/2\rceil-1}\binom{n/2-1}{\lfloor i/2\rfloor-1}\frac{n-i}{i}
   &\leq & \sum_{i=1}^{n-1}\binom{n/2-1}{\lceil i/2\rceil-1}\binom{n/2-1}{\lfloor i/2\rfloor-1}.
\end{eqnarray*}
\end{lem}
\begin{proof}
To prove the above inequality, it is suffices to show that
\begin{eqnarray*}
 && \binom{n/2-1}{\lceil i/2\rceil-1}\binom{n/2-1}{\lfloor i/2\rfloor-1}\frac{n-i}{i}+\binom{n/2-1}{\lceil \frac{n-i}{2}\rceil-1}\binom{n/2-1}{\lfloor \frac{n-i}{2}\rfloor-1}\frac{i}{n-i}
 \\ &\overset{(a)}{\leq}& \binom{n/2-1}{\lceil i/2\rceil-1}\binom{n/2-1}{\lfloor i/2\rfloor-1}+\binom{n/2-1}{\lceil \frac{n-i}{2}\rceil-1}\binom{n/2-1}{\lfloor \frac{n-i}{2}\rfloor-1},
\end{eqnarray*}
for all $1\leq i\leq n/2-1$.

First assume that $i$ is even. Then the inequality (a) is equivalent to
\begin{eqnarray*}
 && \binom{n/2-1}{i/2-1}\binom{n/2-1}{i/2-1}\frac{n-i}{i}+\binom{n/2-1}{i/2}\binom{n/2-1}{i/2}\frac{i}{n-i}
 \\ &\overset{(b)}{\leq}& \binom{n/2-1}{i/2-1}\binom{n/2-1}{i/2-1}+\binom{n/2-1}{i/2}\binom{n/2-1}{i/2},
\end{eqnarray*}
and the inequality (b) holds if and only if
\begin{eqnarray*}
 2M\frac{n-i}{i} &\overset{(c)}{\leq}& M\bigg(1+\bigg(\frac{n-i}{i}\bigg)^2\bigg),
\end{eqnarray*}
where $M:=\binom{n/2-1}{i/2-1}\binom{n/2-1}{i/2-1}.$ Since $M>0$, the inequality (c) holds by letting $x=\frac{n-i}{i}$ in $x^2-2x+1\geq 0$.
A similar argument works in the case when $i$ is odd.
\end{proof}

By Lemma \ref{estimate} we have
\begin{eqnarray*}
  \sum_{\textbf{x}\in V_{n-1}} w(\textbf{x})&\leq & \frac{4}{n-2}\cdot \sum_{i=1}^{n-1}\binom{n/2-1}{\lceil i/2\rceil-1}\binom{n/2-1}{\lfloor i/2\rfloor-1}\\
   &=&  \frac{2\big(\binom{n}{n/2}-2\big)}{n-2},
\end{eqnarray*}
which is an upper bound of $A^{\rm D}(n)$. By Lemma \ref{binomial}, the redundancy $\rho_{b}(n,N;B^{\rm D})\geq n-\log_2(\binom{n}{n/2}-2)+\log_2(n-2)-1=\frac{3}{2}\log_2n-O(1)$.
\end{proof}

\begin{rem}\label{re1}
By the definition of BVT code, we have
$$A^{\rm D}(n)\geq \binom{n}{n/2}/(n+1),\ {\rm and}\ \lim_{n\rightarrow \infty}\frac{A^{\rm D}(n)}{\binom{n}{n/2}/n}\geq 1.$$
Additionally,  Theorem \ref{deletion1} says
$$A^{\rm D}(n)\leq \frac{2\big(\binom{n}{n/2}-2\big)}{n-2},\ {\rm and}\ \lim_{n\rightarrow \infty}\frac{A^{\rm D}(n)}{2\binom{n}{n/2}/n}\leq 1.$$
Hence,
$$\binom{n}{n/2}/n \lesssim A^{\rm D}(n) \lesssim 2\binom{n}{n/2}/n.$$

In fact, following Levenshtein's method in \cite{L65} (or \cite{S00}), we can get a similar but implicit upper bound on  $A^{\rm D}(n)$. However, the constant factor in the estimation is, as yet, unknown. More precisely, determining the constant $1\leq t\leq 2$ such that $\lim_{n\rightarrow \infty}\frac{A^{\rm D}(n)}{t\binom{n}{n/2}/n}=1$, is  challenging.
\end{rem}


By Remark~\ref{re1}, we can determine the value of $\rho_b(n,N;B_2)$ for $N=1$ as follows.

\begin{thm}\label{DI1}
Consider the error-ball $B_2\in \{B^{\rm D},B^{\rm I},B^{\rm DI}\}$. Then
$$\rho_{b}(n,1;B_2)=\frac{3}{2}\log_2n+\Theta(1).$$
\end{thm}

\subsection{The case $N\geq 2$}

%

Chee \emph{et al.} \cite{CKVVY18} defined a special class of binary codes in terms of the period of codewords, which can be used to correct deletions and sticky insertions when the two heads (in racetrack memory) are well separated. We will use this idea to construct balanced reconstruction codes. The following definition is necessary. 

\begin{defe}\label{period}
Let $\ell$ and $m$ be two positive integers with $\ell<m$. Let $\textbf{u}=(u_1,u_2,\ldots, u_m)\in \F_2^m$. We say that the word $\textbf{u}$ has period $\ell$ if $\ell$ is the smallest integer such that  $u_i=u_{i+\ell}$ for all $1\leq i\leq m-\ell$.
\end{defe}

Let $\mathcal{R}_2^b(n,\ell,m)$ denote the set of all binary words $\textbf{c}$ in $U_n$ such that the length of any $\ell'$-periodic ($\ell' \leq \ell$) subword of $\textbf{c}$ is at most $m$. For example, $\mathcal{R}_2^b(6,1,2)$ is the set
$$\{110100,110010,101100,101010,101001,100110,100101,011010,$$
$$011001,010110,010101,010011,001101,001011\}.$$

We are now ready to characterize the size of the set $\mathcal{R}_2^b(n,\ell,m)$ in some particular cases.

\begin{lem}\label{case1}
For all $n,m$ and $\ell=1$, we have
$$|\mathcal{R}_2^b(n,1,m)|\geq \binom{n}{n/2}\bigg(1-n\bigg(\frac{1}{2}\bigg)^m\bigg).$$
In particular, if $m=\lceil \log_2n\rceil+1$, we have $|\mathcal{R}_2^b(n,1,\lceil \log_2n\rceil+1)|\geq \frac{\binom{n}{n/2}}{2}$.
\end{lem}
\begin{proof}
Let $\overline{\mathcal{R}_2^b(n,1,m)}$ denote the complementary set $U_n\backslash \mathcal{R}_2^b(n,1,m)$, then $|\mathcal{R}_2^b(n,1,m)| = |U_n|-|\overline{\mathcal{R}_2^b(n,1,m)}|$.
By  definition, a word $\textbf{c}\in U_n$ belongs to $\overline{\mathcal{R}_2^b(n,1,m)}$ if and only if it contains a run with length $m+1$, which implies the following upper bound on the size of $\overline{\mathcal{R}_2^b(n,1,m)}$
\begin{eqnarray*}
  |\overline{\mathcal{R}_2^b(n,1,m)}| &\leq& 2(n-m)\binom{n-m-1}{n/2}=(n-2m)\binom{n-m}{n/2}.
\end{eqnarray*}
Then it follows that
\begin{eqnarray*}
  |\mathcal{R}_2^b(n,1,m)| &\geq& \binom{n}{n/2}-(n-2m)\binom{n-m}{n/2} \\
   &\overset{(a)}{\geq}& \binom{n}{n/2}\bigg(1-(n-2m)\bigg(\frac{1}{2}\bigg)^m\bigg) \\
   &\geq& \binom{n}{n/2}\bigg(1-n\bigg(\frac{1}{2}\bigg)^m\bigg),
\end{eqnarray*}
where the inequality (a) follows from the fact that $\binom{n-m}{n/2}\leq \binom{n}{n/2}\big(\frac{1}{2}\big)^m$.
\end{proof}

\begin{lem}\label{case2}
For all $n,m$ and $\ell=2$, we have
$$|\mathcal{R}_2^b(n,2,m)|\geq \binom{n}{n/2}\bigg(1-n\bigg(\frac{1}{2}\bigg)^{\lceil \frac{m-1}{2}\rceil}\bigg).$$
In particular, if $n\geq 12$ and $m=2\lceil \log_2n\rceil+3$, we have $|\mathcal{R}_2^b(n,2,2\lceil \log_2n\rceil+3)|\geq \frac{\binom{n}{n/2}}{2}$.
\end{lem}
\begin{proof}
Let
$\mathcal{B}_2(n,2,m)$ denote the set of all binary words $\textbf{c}$ in $U_n$ such that the length of any $2$-periodic subword of $\textbf{c}$ is at most $m$. Then $\mathcal{R}_2^b(n,2,m)=\mathcal{B}_2(n,2,m)\cap \mathcal{R}_2^b(n,1,m)$, thus
\begin{eqnarray*}
  |\mathcal{R}_2^b(n,2,m)| &\geq & |\mathcal{R}_2^b(n,1,m)|-|\overline{\mathcal{B}_2(n,2,m)}|,
\end{eqnarray*}
where $\overline{\mathcal{B}_2(n,2,m)}$ denotes the complementary set $U_n\backslash \mathcal{B}_2(n,2,m)$.
Note that a word $\textbf{u}\in U_n$ belongs to $\overline{\mathcal{B}_2(n,2,m)}$ if and only if it contains a subword of length $m+1$ with period 2.
Hence,
\begin{eqnarray*}
  |\overline{\mathcal{B}_2(n,2,m)}| &\leq& 2(n-m)\binom{n-m-1}{\frac{n}{2}-\lceil \frac{m+1}{2}\rceil}
   \leq n\binom{n-m}{\frac{n}{2}-\lceil \frac{m+1}{2}\rceil}.
\end{eqnarray*}
By Lemma \ref{case1},
\begin{eqnarray*}
  |\mathcal{R}_2^b(n,2,m)| &\geq& \binom{n}{n/2}\bigg(1-n\bigg(\frac{1}{2}\bigg)^m\bigg)-n\binom{n-m}{\frac{n}{2}-\lceil \frac{m+1}{2}\rceil} \\
   &\overset{(a)}{\geq}& \binom{n}{n/2}\bigg(1-n\bigg(\bigg(\frac{1}{2}\bigg)^m+\bigg(\frac{1}{2}\bigg)^{\lceil \frac{m+1}{2}\rceil}\bigg)\bigg) \\
   &\geq& \binom{n}{n/2}\bigg(1-n\bigg(\frac{1}{2}\bigg)^{\lceil \frac{m-1}{2}\rceil}\bigg),
\end{eqnarray*}
where the inequality (a) follows from the fact that
\begin{eqnarray*}
  \binom{n-m}{\frac{n}{2}-\lceil \frac{m+1}{2}\rceil} &\leq& \binom{n}{n/2}\bigg(\frac{1}{2}\bigg)^{\lceil \frac{m+1}{2}\rceil}\bigg(\frac{\frac{n}{2}}{n-\lceil \frac{m+1}{2}\rceil}\bigg)^{\lfloor \frac{m-1}{2}\rfloor} \\
   &\leq&  \binom{n}{n/2}\bigg(\frac{1}{2}\bigg)^{\lceil \frac{m+1}{2}\rceil}.
\end{eqnarray*}
If $m=2\lceil \log_2n\rceil+3$, we have
$$|\mathcal{R}_2^b(n,2,2\lceil \log_2n\rceil+3)|\geq \binom{n}{n/2}\bigg(1-n\bigg(\frac{1}{2}\bigg)^{\log_2n+1}\bigg)=\frac{\binom{n}{n/2}}{2},$$
and the condition $m\leq n$ holds as long as $n\geq 12$.
\end{proof}

For any $\textbf{x}=(x_1,x_2,\ldots,x_n)\in \F_2^n$, define the  inversion number
$${\rm Inv}(\textbf{x})=|\{(i,j):1\leq i<j\leq n, x_i>x_j\}|.$$
For example, ${\rm Inv}(\textbf{x})=7$ for $\textbf{x}=(1010110)\in \F_2^7$.

Based on Lemma \ref{case2}, we give the following estimate of the optimal redundancy of an $(n,N;B_2)$-reconstruction code, with $B_2\in \{B^{\rm D},B^{\rm I}\}$ and $N\geq 2$.

\begin{thm}\label{D-Iredundancy}
Consider the error-ball $B_2\in \{B^{\rm D},B^{\rm I}\}$. Then
\[
\rho_{b}(n,N;B_2)=\begin{cases}
\frac{1}{2}\log_2n+\log_2\log_2n+\Theta(1), & N=2, \\
\Delta, & N\geq 3.
\end{cases}
\]
\end{thm}
\begin{proof}
The case $N\geq 3$ is trivial. Let
$$\mathcal{C}_2(n,t,P)=\{\textbf{c}\in \mathcal{R}_2^b(n,2,P): {\rm Inv}(\textbf{c})\equiv t\pmod {1+P/2}\},$$
where $t\in \Z_{1+P/2}$ and $P$ is even. It follows from Proposition \ref{D-Icondition}(ii) and \cite[Theorem 17]{CKNY21} that
$$|B_2(\textbf{x})\cap B_2(\textbf{y})|<2,\ {\rm for}\ \textbf{x}\neq\textbf{y}\in \mathcal{C}_2(n,t,P).$$
Assume that $P=2\lceil \log_2n\rceil+3$. Then by Lemma \ref{case2}, there exists an $(n,2;B_2)$-reconstruction code with suitable $t$ such that
$$|\mathcal{C}_2(n,t,2\lceil \log_2n\rceil+3)|\geq \binom{n}{n/2}/(2+P).$$
This implies by Lemma \ref{binomial}, that $\mathcal{C}_2(n,t,2\lceil \log_2n\rceil+3)$ has redundancy at most
$\frac{1}{2}\log_2n+\log_2(2+P)+\frac{1}{2}=\frac{1}{2}\log_2n+\log_2\log_2n+O(1)$.

On the other hand, let $\mathcal{C}$ be any $(n,2;B_2)$-reconstruction code. By Proposition \ref{D-Icondition}(ii), every pair of different words in $\mathcal{C}$ are not Type-A-confusable (see Definition \ref{confusableAB}). Then the desired result follows immediately from Theorem \ref{lowerbound}(i).
\end{proof}

A further extension of Theorem \ref{D-Iredundancy} is given by the following theorem.

\begin{thm}\label{DIredundancy}
Consider the error-ball $B^{\rm DI}$. Then
\[
\rho_{b}(n,N;B^{\rm DI})=\begin{cases}
\frac{3}{2}\log_2n+\Theta(1), & N=2, \\
\frac{1}{2}\log_2n+\log_2\log_2n+\Theta(1), & N\in\{3,4\},\\
\Delta, & N\geq 5.
\end{cases}
\]
\end{thm}
\begin{proof}
The case $N\geq 5$ is trivial. Since $|B^{\rm DI}(\textbf{x})\cap B^{\rm DI}(\textbf{y})|\in\{0,2,4\}$ for distinct words $\textbf{x},\textbf{y}\in U_n$ by Corollary \ref{DIcondition}, we have $\rho_{b}(n,2;B^{\rm DI})=\rho_{b}(n,1;B^{\rm DI})$ and $\rho_{b}(n,4;B^{\rm DI})=\rho_{b}(n,3;B^{\rm DI})$ directly. Thus, the value $\rho_{b}(n,2;B^{\rm DI})$ follows from Theorem \ref{DI1}. In addition, the proof of Theorem \ref{D-Iredundancy} shows that the code $\mathcal{C}_2(n,t,2\lceil \log_2n\rceil+3)$ is in fact an $(n,4;B^{\rm DI})$-reconstruction code, and the proof is complete.
\end{proof}

\section{Reconstruction codes with error-balls $B^{\rm SD}$, $B^{\rm SI}$ and $B^{\rm edit}$}\label{Sec:4}

In connection of the preceding discussion, for instance, see Theorems \ref{substitution} and \ref{DI1}, we mention without proof the following result for the optimal redundancy of an $(n,N;B_2)$-reconstruction code, where $B_2\in\{B^{\rm SD},B^{\rm SI}\}$.

\begin{coro}\label{SDI1}
Consider the error-ball $B_2\in\{B^{\rm SD},B^{\rm SI}\}$. Then
$$\rho_{b}(n,N;B_2)=\frac{3}{2}\log_2n+\Theta(1)\ {\rm for}\ N\in\{1,2\}.$$
\end{coro}

We are now in a position to evaluate the value $\rho_{b}(n,N;B_2)$ for $B_2\in\{B^{\rm SD},B^{\rm SI}\}$ when $N\geq 3$.

\begin{thm}\label{SDIredundancy}
Consider the error-ball $B_2\in\{B^{\rm SD},B^{\rm SI}\}$. Then
\[
\rho_{b}(n,N;B_2)=\begin{cases}
\frac{1}{2}\log_2n+\log_2\log_2n+\Theta(1), & N=3,\\
\Delta+1-o(1), & N=4,\\
\Delta, & N\geq 5.
\end{cases}
\]
\end{thm}
\begin{proof}
Assume that
$$\mathcal{D}_2(n,t,P)=\{\textbf{c}\in \mathcal{R}_2^b(n,1,P): {\rm Inv}(\textbf{c})\equiv t\pmod {1+P}\},$$
where $t\in \Z_{1+P}$. It follows then from Proposition \ref{SDIcondition}(i) that
$$|B_2(\textbf{x})\cap B_2(\textbf{y})|<3,\ {\rm for}\ \textbf{x}\neq\textbf{y}\in \mathcal{D}_2(n,t,P).$$
If $P=\lceil \log_2n\rceil+1$. Then we have
$$|\mathcal{D}_2(n,t,\lceil \log_2n\rceil+1)|\geq \binom{n}{n/2}/(2+2P).$$
This implies by Theorem \ref{binomial} and Lemma \ref{case1}, that $\mathcal{D}_2(n,t,\lceil \log_2n\rceil+1)$ has redundancy at most
$\frac{1}{2}\log_2n+\log_2\log_2n+O(1)$.
On the other hand, let $\mathcal{C}$ be an $(n,3;B_2)$-reconstruction code.
It is easy to check that every pair of different words in $\mathcal{C}$ are not Type-B-confusable (see Definition \ref{confusableAB}). Combining this and Theorem \ref{lowerbound}(ii), we give the result for $N=3$.

In the case of $N=4$. Define the set
$$\mathcal{C}_a=\{(x_1,x_2,\ldots,x_n)\in U_n:\sum_{i=1}^{n/2}x_{2i}\equiv a\pmod 2\},$$
where $a\in \Z_2$. Then the pigeonhole principle implies that there is a choice of $a\in \Z_2$ such that the set $\mathcal{C}_a$ has size at lease half of $U_n$. By Proposition \ref{SDIcondition}(i),
$$B_2(\textbf{x})\cap B_2(\textbf{y})<4,\ {\rm for}\ \textbf{x}\neq\textbf{y}\in C_a.$$
Thus, $\mathcal{C}_a$ is an $(n,4;B_2)$-reconstruction code with redundancy at most $n-\log_2\frac{\binom{n}{n/2}}{2}=\Delta+1$.
On the other hand, every distinct pair of words in an $(n,4;B_2)$-reconstruction code are not Type-B-confusable with $m=1$; otherwise the two words are Type-A-confusable with $m=1$, and Proposition \ref{SDIcondition} indicates that the size of their error-balls equals 4, a contradiction. Then Theorem \ref{lowerbound}(iii) implies the desired result.
\end{proof}

Next, we consider the error-ball $B^{\rm edit}$. First, we define the balanced version of the Levenshtein code proposed in \cite{L65} as follows, which is a generalization of Definition \ref{restrictedVT},
$$BLT_a(n)=\{(x_1,x_2,\ldots,x_n)\in U_n:\sum_{i=1}^nix_i \equiv a\pmod {2n}\}.$$
In \cite{L65}, Levenshtein showed that the code $BLT_a(n)$ is capable of correcting one deletion, one insertion or one substitution. So there is a choice of $a\in \Z_{2n}$  such that $|BLT_a(n)|\geq \frac{\binom{n}{n/2}}{2n}$. This leads to half of the following theorem.

\begin{thm}\label{E1}
Consider the error-ball $B^{\rm edit}$. We have that
$$\rho_{b}(n,N;B^{\rm edit})=\frac{3}{2}\log_2n+\Theta(1)\ {\rm for}\ N\in \{1,2\}.$$
\end{thm}
\begin{proof}
Recall from Proposition \ref{Econdition} that $\rho_{b}(n,1;B^{\rm edit})=\rho_{b}(n,2;B^{\rm edit})$. Additionally, by the code $BLT_a(n)$ defined above, it suffices to show that the value $\rho_{b}(n,N;B^{\rm edit})$ is lower bounded by $\frac{3}{2}\log_2n+\Theta(1)$. Notice that an $(n,1;B^{\rm edit})$-reconstruction code is also an $(n,1;B_2)$-reconstruction code with $B_2\in \{B^{\rm S},B^{\rm D},B^{\rm I},B^{\rm SD},B^{\rm SI}\}$, and then the theorem follows.
\end{proof}

We now complete the the evaluation of  of $\rho_{b}(n,N;B^{\rm edit})$ with $N\geq 3$.

\begin{thm}\label{DIredundancy}
Consider the error-ball $B^{\rm edit}$. Then
\[
\rho_{b}(n,N;B^{\rm edit})=\begin{cases}
\frac{1}{2}\log_2n+\log_2\log_2n+\Theta(1), & N\in \{3,4\},\\
\Delta+1-o(1), & N\in \{5,6\},\\
\Delta, & N\geq 7.
\end{cases}
\]
\end{thm}
\begin{proof}
The case $N\geq 7$ is trivial. By Proposition \ref{Econdition},
$$\rho_{b}(n,4;B^{\rm edit})=\rho_{b}(n,3;B^{\rm edit})\ {\rm and}\ \rho_{b}(n,6;B^{\rm edit})=\rho_{b}(n,5;B^{\rm edit}).$$
For $N\in \{3,4\}$, suppose that
$$\mathcal{E}_2(n,t,P)=\{\textbf{c}\in \mathcal{R}_2^b(n,2,P): {\rm Inv}(\textbf{c})\equiv t\pmod {1+P}\},$$
where $t\in \Z_{1+P}$. Clearly, we have
$$|B^{\rm edit}(\textbf{x})\cap B^{\rm edit}(\textbf{y})|<4\ {\rm(in\ fact \leq 2)},\ {\rm for}\ \textbf{x}\neq\textbf{y}\in \mathcal{E}_2(n,t,P).$$
Assume that $P=2\lceil \log_2n\rceil+3$. Then we have $|\mathcal{E}_2(n,t,\lceil 2\log_2n\rceil+3)|\geq \binom{n}{n/2}/(2+2P)$.
This implies by Lemmas \ref{binomial} and \ref{case2}, that $\mathcal{E}_2(n,t,2\lceil \log_2n\rceil+3)$ has redundancy at most
$\frac{1}{2}\log_2n+\log_2\log_2n+O(1)$. On the other hand, let $\mathcal{C}$ be an $(n,N;B^{\rm edit})$-reconstruction code.
It is easy to check that every pair of different words in $\mathcal{C}$ are not Type-B-confusable. Consequently, we have $\rho_{b}(n,N;B^{\rm edit})=\frac{1}{2}\log_2n+\log_2\log_2n+\Theta(1)$.

In the case of $N\in\{5,6\}$, the code $\mathcal{C}_a$ defined in the proof of Theorem \ref{SDIredundancy} is exactly an $(n,N;B^{\rm edit})$-reconstruction code with redundancy at most $n-\log_2\frac{\binom{n}{n/2}}{2}=\Delta+1$. Combining this and Theorem \ref{lowerbound}(iii), we obtain the value $\rho_{b}(n,N;B^{\rm edit})$ immediately.
\end{proof}

\section{Conclusion}\label{Sec:5}

In this paper, we completely determine the asymptotic optimal redundancy for  balanced binary reconstruction codes which are affected by  single edits, i.e., one substitution, one deletion, one insertion and their combinations. It is interesting to notice that for all possible single edits, the redundancy of an asymptotically optimal balanced reconstruction code gradually decreases from $\frac{3}{2}\log _2 n+O(1)$ to $\frac{1}{2}\log_2n+\log_2\log_2n+O(1)$, and finally to $\frac{1}{2}\log_2n+O(1)$ but with different speeds.

 Because of the balanced property, the optimal redundancy is not surprisingly bigger than the unbalanced one studied in \cite{CKNY21} with the same noisy channel. However, if we define
$$\rho'_b(n,N;B_2)=\min\{(n-\Delta)-\log_2 |\mathcal{C}|: \mathcal{C}\subseteq U_n \ {\rm and} \  \nu(\mathcal{C};B_2)<N\},$$ where $\Delta$ is the redundancy of $U_n$,
then the asymptotical redundancy here is consistent with \cite{CKNY21} in binary case. This in turn implies that the balanced constraint does not reduce the proportion of the $(n,N;B_2)$-reconstruction code in the corresponding codebook.

Moreover, it would be interesting to investigate the sequence reconstruction problem constrained in the balanced quaternary sequences, for instance, see \cite{TD10,YT14} for a description of this family of sequences. The case of noisy channel with $t$-deletion (insertion) error-balls will be considered as another potential path for further work.

\vspace*{0.5cm}

\hspace*{-0.6cm}\textbf{Data availibility} Not applicable.\\

\hspace*{-0.6cm}\textbf{Code Availability} Not applicable.\\

\hspace*{-0.6cm}\textbf{\Large Declarations}\\

\hspace*{-0.6cm}\textbf{Conflict of interest} The authors have no conflicts of interest to declare that are relevant to the content of this paper.

\end{document}